\documentclass[11pt]{article}
\usepackage[  textheight=9in, textwidth=6.5in, top=1in, headheight=14pt, headsep=25pt, footskip=30pt]{geometry}
\usepackage{amsmath,amssymb,amsfonts,mathtools}
\usepackage{bbm}
\usepackage{color,graphicx}
\usepackage{verbatim}
\usepackage{pgf,tikz}
\usepackage[linesnumbered,ruled,vlined]{algorithm2e}
\usepackage{algorithmic}
\usetikzlibrary{arrows}
\usepackage{capt-of}
\usepackage{setspace}
\usepackage{caption}
\usepackage{subcaption}
\usepackage{fancyhdr}
\usepackage[colorinlistoftodos]{todonotes}
\usepackage{rotating}
\usepackage{booktabs}
\usepackage{enumitem, comment}
\usepackage{soul} 
\usepackage{dsfont}
\usepackage{bookmark}
\usepackage[utf8]{inputenc} %
\usepackage{csquotes}
\usepackage{fnpct} %
\usepackage{float}
\usepackage[utf8]{inputenc}
\usepackage[english]{babel}
\usepackage{parskip}

\usepackage[T1]{fontenc}
\usepackage{lmodern}
\usepackage{xspace}    
\usepackage[protrusion=true,expansion=true]{microtype}  

\usepackage{amsthm}
\usepackage{thm-restate}
\usepackage{ifdraft}

\usepackage{nicefrac}

\usepackage{mdframed}
\newmdtheoremenv{theomd}{Theorem}

\newtheorem{theorem}{Theorem}[section]
\newtheorem*{theorem*}{Theorem}

\newtheorem*{claim*}{Claim}

\newtheorem*{proposition*}{Proposition}
\newtheorem{lemma}[theorem]{Lemma}
\newtheorem*{lemma*}{Lemma}

\newtheorem*{conjecture*}{Conjecture}

\newtheorem{fact}[theorem]{Fact}
\newtheorem*{fact*}{Fact}

\newtheorem*{hypothesis*}{Hypothesis}
\theoremstyle{plain}
\newtheorem{definition}[theorem]{Definition}
\newmdtheoremenv{defmd}[theorem]{Definition}

\newtheorem{problem}[theorem]{Problem}

\newtheorem{remark}[theorem]{Remark}

\newtheorem{question}[theorem]{Question}

\usepackage{prettyref}
\newcommand{\savehyperref}[2]{\texorpdfstring{\hyperref[#1]{#2}}{#2}}

\newrefformat{eq}{\savehyperref{#1}{\textup{(\ref*{#1})}}}
\newrefformat{lem}{\savehyperref{#1}{Lemma~\ref*{#1}}}
\newrefformat{def}{\savehyperref{#1}{Definition~\ref*{#1}}}
\newrefformat{thm}{\savehyperref{#1}{Theorem~\ref*{#1}}}
\newrefformat{cor}{\savehyperref{#1}{Corollary~\ref*{#1}}}
\newrefformat{cha}{\savehyperref{#1}{Chapter~\ref*{#1}}}
\newrefformat{sec}{\savehyperref{#1}{Section~\ref*{#1}}}
\newrefformat{app}{\savehyperref{#1}{Appendix~\ref*{#1}}}
\newrefformat{tab}{\savehyperref{#1}{Table~\ref*{#1}}}
\newrefformat{fig}{\savehyperref{#1}{Figure~\ref*{#1}}}
\newrefformat{hyp}{\savehyperref{#1}{Hypothesis~\ref*{#1}}}
\newrefformat{alg}{\savehyperref{#1}{Algorithm~\ref*{#1}}}
\newrefformat{sdp}{\savehyperref{#1}{SDP~\ref*{#1}}}
\newrefformat{rmk}{\savehyperref{#1}{Remark~\ref*{#1}}}
\newrefformat{item}{\savehyperref{#1}{Item~\ref*{#1}}}
\newrefformat{step}{\savehyperref{#1}{step~\ref*{#1}}}
\newrefformat{conj}{\savehyperref{#1}{Conjecture~\ref*{#1}}}
\newrefformat{fact}{\savehyperref{#1}{Fact~\ref*{#1}}}
\newrefformat{prop}{\savehyperref{#1}{Proposition~\ref*{#1}}}
\newrefformat{claim}{\savehyperref{#1}{Claim~\ref*{#1}}}
\newrefformat{relax}{\savehyperref{#1}{Relaxation~\ref*{#1}}}
\newrefformat{red}{\savehyperref{#1}{Reduction~\ref*{#1}}}
\newrefformat{part}{\savehyperref{#1}{Part~\ref*{#1}}}
\newrefformat{prob}{\savehyperref{#1}{Problem~\ref*{#1}}}
\newrefformat{ass}{\savehyperref{#1}{Assumption~\ref*{#1}}}
\newrefformat{ques}{\savehyperref{#1}{Question~\ref*{#1}}}

\newrefformat{ex}{\savehyperref{#1}{Example~\ref*{#1}}}
\usepackage{bm}

\hypersetup{colorlinks,linkcolor=blue,filecolor = blue, citecolor = violet, urlcolor  = magenta}

\usepackage[backend=biber, style=alphabetic,
sorting=nyt, natbib=true, backref=false, eprint=true, url=true, doi = true, isbn = false, maxbibnames = 99, maxcitenames = 2, maxalphanames = 4]{biblatex}
\DeclareSourcemap{
  \maps[datatype=bibtex]{
    \map[overwrite]{
      \step[fieldsource=eprint, final]
      \step[fieldset=doi, null]
    }
  }
}

\DeclareSourcemap{
  \maps[datatype=bibtex]{
    \map[overwrite]{
      \step[fieldsource=doi, final]
      \step[fieldset=url, null]
    }
  }
}

\DeclareSourcemap{
  \maps[datatype=bibtex]{
    \map[overwrite]{
      \step[fieldsource=eprint, final]
      \step[fieldset=url, null]
    }
  }
}

\DeclareSourcemap{
  \maps[datatype=bibtex, overwrite]{
    \map{
      \step[fieldset=address, null]
      \step[fieldset=location, null]
    }
  }
}

\renewenvironment{abstract}
  {{\centering\large\bfseries Abstract\par}\vspace{0.7ex}%
    \bgroup
       \leftskip 20pt\rightskip 20pt\small\noindent\ignorespaces}%
  {\par\egroup\vskip 0.25ex}

\newenvironment{keywords}
{\bgroup\leftskip 20pt\rightskip 20pt \small\noindent{\bfseries
Keywords:} \ignorespaces}%
{\par\egroup\vskip 0.25ex}
\newlength\aftertitskip     \newlength\beforetitskip
\newlength\interauthorskip  \newlength\aftermaketitskip

\newcommand{\mper}{\,.}
\newcommand{\mcom}{\,,}

\newcommand{\paren}[1]{\left(#1 \right )}

\newcommand{\Brac}[1]{\left[#1\right]}

\newcommand{\Set}[1]{\left\{#1\right\}}

\newcommand{\abs}[1]{\left\lvert#1\right\rvert}
\newcommand{\Abs}[1]{\left\lvert#1\right\rvert}

\newcommand{\norm}[1]{\left\lVert#1\right\rVert}

\newcommand{\defeq}{\stackrel{\textup{def}}{=}}

\newcommand{\N}{{\mathbb Z}_{\geq 0}}
\newcommand{\R}{\mathbb R}

\newcommand{\Psymb}{\mathbb{P}}

\DeclareMathOperator*{\ProbOp}{\Psymb}

\newcommand{\ber}{\mathsf{Ber}}

\renewcommand{\Pr}[1]{\ProbOp\Brac{#1}}

\newcommand{\e}{\epsilon}

\definecolor{DSgray}{cmyk}{0,0,0,0.7}

\let\e\varepsilon

\newcommand{\cE}{\mathcal E}

\newcommand{\poly}{{\sf poly}}
\newcommand{\polylog}{{\sf polylog}}

\newcommand{\rd}{{\sf d}}

\newcommand{\tr}{\mathrm{tr}}

\newcommand{\tOnaran}{\widetilde{\mathcal{O}}_{\mathsf{NA-RW}}}

\usepackage{dsfont}
\usepackage{bbm}

\newcommand{\blambda}{\bm{\lambda}}
\newcommand{\bLambda}{\bm{\Lambda}}

\let\epsilon\varepsilon
  \usepackage{nth}
  \usepackage{intcalc}

  \newcommand{\cAAAI}[1]{AAAI\ Conference\ on\ Artificial (AAAI)}

   \title{Moments, Random Walks, and Limits for Spectrum Approximation}
\author{Yujia Jin \\ Stanford University \\ \texttt{yujia@stanford.edu} \and Christopher Musco \\ New York University \\ \texttt{cmusco@nyu.edu} \and Aaron Sidford \\ Stanford University \\ \texttt{sidford@stanford.edu} \and  Apoorv Vikram Singh \\ New York University \\ \texttt{apoorv.singh@nyu.edu}}
\date{}
\begin{document}
\maketitle

\begin{abstract}%
	We study lower bounds for the problem of approximating a one dimensional distribution given (noisy) measurements of its moments. We show that there are distributions on $[-1,1]$ that cannot be approximated to accuracy $\epsilon$ in Wasserstein-1 distance even if we know \emph{all} of their moments to multiplicative accuracy $(1\pm2^{-\Omega(1/\epsilon)})$; this result matches an upper bound of Kong and Valiant [Annals of Statistics, 2017].  To obtain our result, we provide a hard instance involving distributions induced by the eigenvalue spectra of carefully constructed graph adjacency matrices. Efficiently approximating such spectra in Wasserstein-1 distance is a well-studied algorithmic problem, and a recent result of Cohen-Steiner et al. [KDD 2018] gives a method based on accurately approximating spectral moments using $2^{O(1/\epsilon)}$ random walks initiated at uniformly random nodes in the graph.

    As a strengthening of our main result, we show that improving the dependence on $1/\epsilon$ in this result would require a new algorithmic approach. Specifically, no algorithm can compute an $\epsilon$-accurate approximation to the spectrum of a normalized graph adjacency matrix with constant probability, even when given the transcript of $2^{\Omega(1/\epsilon)}$ random walks of length $2^{\Omega(1/\epsilon)}$ started at random nodes.
\end{abstract}


\section{Introduction}

A fundamental problem in linear algebra is to approximate the full list of eigenvalues, $\lambda_1 \leq \ldots \leq \lambda_n \in \R$, of a symmetric matrix $A \in \R^{n \times n}$,  ideally in less time than it takes to compute a full eigendecomposition.\footnote{All eigenvalues can be computed to precision $\epsilon$ in ${O}(n^{\omega + \eta}\polylog(\tfrac{n}{\epsilon}))$ time, where $\omega \approx 2.373$ is the  matrix multiplication constant  \citep{BanksGarza-VargasKulkarni:2022}.  Methods typically used in practice run in time $O(n^3 + n^2\log(\tfrac{1}{\epsilon}))$ \citep{Wilkinson:1968}.} We focus on the particular problem of \emph{spectral density estimation} where given $\epsilon \in (0, 1)$ and the assumption that $\norm{A}_2 \leq 1$, the goal is find approximate eigenvalues ${\lambda'_1} \leq \ldots \leq {\lambda'_n}$ such that their average absolute error is bounded by $\epsilon$, i.e.,
\begin{align}
	\label{eq:main_guar}
	\frac{1}{n}\sum_{i=1}^n |\lambda_i - \lambda'_i| \leq \epsilon.
\end{align}
This problem is equivalent to that of computing an $\epsilon$-approximation in Wasserstein-1 distance to the distribution on $[-1,1]$ induced by the \emph{spectral density (function)} of $A$, i.e.\ $p(x) \defeq \frac{1}{n}\sum_{i=1}^n \delta(x-\lambda_i)$ for indicator function $\delta$ (see \prettyref{sec:prelims} for notation).

Spectral density estimation is distinct from and in many ways more challenging than related problems like low-rank approximation, where we only seek to approximate the \emph{largest magnitude} eigenvalues of $A$. Nevertheless, efficient randomized algorithms for spectral density estimation were developed in the early 1990s and have been applied widely in computational physics and chemistry \citep{Skilling:1989,SilverRoder:1994,Wang:1994,WeisseWelleinAlvermann:2006}. These algorithms, which include the kernel polynomial and stochastic Lanczos quadrature methods,
 achieve $\epsilon$ accuracy with high probability in roughly $O(n^2/\epsilon)$ time, improving on the  $\Omega(n^{\omega})$ cost of a full eigendecomposition for moderate values of $\epsilon$  \citep{ChenTrogdonUbaru:2021}.

More recently, there has been a resurgence of interest in spectral density estimation within the machine learning and data science communities. Research activity in this area has been fueled by emerging applications in analyzing and understanding deep neural networks \citep{PenningtonSchoenholzGanguli:2018,MahoneyMartin:2019, Papyan:2018}, in optimization \citep{GhorbaniKrishnanXiao:2019,SagunEvciGuney:2017}, and in network science \citep{DongBensonBindel:2019,Cohen-SteinerKongSohler:2018}.

\subsection{Spectral Density Estimation for Graphs}
Interestingly, when $A$ is the normalized adjacency matrix\footnote{
	If $\tilde{A}$ is the unnormalized adjacency matrix of $G$ and $D$ is its diagonal degree matrix, we can equivalently consider the asymmetric matrix, $D^{-1}\tilde{A}$ or the symmetric one, $D^{-1/2}\tilde{A}D^{-1/2}$, as they have the same eigenvalues.}
 of an undirected graph $G$, there are faster spectral density estimation algorithms than for general matrices. Specifically, assume that we can randomly sample a node from $G$ and, given a node, randomly sample a neighbor, both in $O(1)$ time. This is possible, for example, in the word RAM model when given arrays containing the neighbors for each node in $G$, and is also a commonly assumed access for computing on extremely large implicit networks \citep{KatzirLibertySomekh:2011}.
 It was recently shown that the $O(n^2/\epsilon)$ runtime of general purpose algorithms like stochastic Lanczos quadrature can be improved to $\tilde{O}(n/\poly(\epsilon))$ \citep{BravermanKrishnanMusco:2022}.\footnote{We use $\tilde{O}(m)$ to denote $O(m\log m)$. The runtime in \citep{BravermanKrishnanMusco:2022} can be improved by a logarithmic factor to ${O}(n/\poly(\epsilon))$ if we have access to a precomputed list of the degrees of nodes in $G$.} This runtime is sublinear in the size of $A$, e.g., when the matrix has $\Omega(n^2)$ non-zero entries.

 Perhaps even more surprisingly, it is possible to solve spectral density estimation for normalized adjacency matrices without any dependence on $n$. Suppose that we are given a \emph{weighted} graph $G$, and again that we can randomly sample a node from $G$ in $O(1)$ time. Also assume that, for any given node, we can randomly sample a neighbor with probability proportional to its edge weight in $O(1)$ time.
In other words, we can initialize and take steps of an edge-weighted random walk in $G$ in $O(1)$ time.\footnote{To be more concrete, if a node $x$ is connected to neighbors $y_1, \ldots, y_d$ with edge weights $w_1, \ldots, w_d$, then the walk steps from $x$ to $y_i$ with probability $w_i/\sum_j w_j$.}
 Then
 \citet{Cohen-SteinerKongSohler:2018} gives an algorithm for any weighted undirected graph that solves the spectral density estimation problem with high probabilty in $2^{O(1/\epsilon)}$ time\footnote{Note that \citet{Cohen-SteinerKongSohler:2018} output a list of approximation eigenvalues $\lambda'_1, \ldots, \lambda'_n$ with only $O(1/\epsilon)$ distinct values that can be stored and returned in time independent of $n$.}.  While completely independent of the graph size, the poor dependence on $\epsilon$ in the result of  \citet{Cohen-SteinerKongSohler:2018} unfortunately makes the algorithm impractical for any reasonable level of accuracy. As such, an interesting question is whether the exponential dependence on $\epsilon$ can be improved (maybe even to polynomial), while still avoiding any dependence on the graph size $n$.
\begin{question}
	\label{ques:question1}
 Can we solve the spectral density estimation problem for a normalized adjacency matrix $A$ given access to $2^{o(1/\epsilon)}$ steps of random walks in the associated graph?
\end{question}

Central to this question is the connection between spectral density estimation and the problem of learning a one dimensional distribution $p$ given noisy measurements of $p$'s (raw) moments. In this work, we consider distributions supported on the the $[-1,1]$, in which case these moments are:
\begin{align*}
	\int_{-1}^1 x p(x) dx, \hspace{.5em}	\int_{-1}^1 x^2 p(x) dx, \hspace{.5em}\int_{-1}^1 x^3p(x) dx, \ldots \mper
\end{align*}
Recent work of \citet{KongValiant:2017} shows that, for a fixed constant $c$, if the first $\ell = c/\epsilon$ moments of any two distributions $p$ and $q$ supported on $[-1,1]$ match \emph{exactly}, then the Wasserstein-1 distance between those distributions is {at most} $\epsilon$.
Given that the left hand size of~\eqref{eq:main_guar} exactly equals the Wasserstein-1  distance $W_1(p,q)$ between the discrete distributions $p(x) = \frac{1}{n}\sum_{i=1}^n \delta(x-\lambda_i)$ and $q(x) = \frac{1}{n}\sum_{i=1}^n \delta(x-\lambda_i')$, the approach in \citet{Cohen-SteinerKongSohler:2018} is to approximate the first $\ell$ moments of $p$, and then to find a set of approximate eigenvalues and eigenvalue multiplicities that correspond to a discrete distribution $q$ with the same moments. Given the approximate moments, finding $q$ can be done in $\poly(\ell)$ time using linear programming algorithms.

Computing the estimates of $p$'s moments is more challenging. \citet{Cohen-SteinerKongSohler:2018} take advantage of the fact that for any $j \leq \ell$, the $j^\text{th}$ moment of $p$ is equal to $\frac{1}{n}\sum_{i=1}^n \lambda_i^j = \frac{1}{n}\tr(A^j)$. This trace can in turn be estimated by random walks of length $j$ in $A$: if we start a random walk at a random node $v$, the probability that we return to $v$ at the $j^\text{th}$ step is exactly equal to $\frac{1}{n}\tr(A^j)$. So, we can obtain an unbiased estimate for the $j^\text{th}$ moment by simply running random walks from random starting nodes and calculating the empirical frequency that we return to our starting point.

This approach leads to the remarkably simple algorithm of \citet{Cohen-SteinerKongSohler:2018}. So where does the $2^{O(1/\epsilon)}$ runtime dependence come from? The issue is that the result of \citet{KongValiant:2017} is brittle to noise. In particular, if the sum of squared distances between $p$'s moments and $q$'s moments differ by $\Delta$, the bound from \citet{KongValiant:2017} weakens, only showing that the Wasserstein-1 distance is bounded by $O(\frac{1}{\ell} + \Delta \cdot 3^\ell)$. To obtain accuracy $\epsilon$, it is necessary to set $\ell = O(1/\epsilon)$ and thus $\Delta$ equal to $2^{-O(1/\epsilon)}$. By standard concentration inequalities, to obtain such an accurate estimate to $p$'s moments, we need to run an exponential number of random walks of length $1, \ldots, \ell$. Accordingly, an important step towards answering \prettyref{ques:question1}   is to understand if such extremely accurate estimates of the moments is necessary for spectral density estimation.

Note that many other spectral density estimation algorithms for general matrices are also based on {moment}-matching. A common approach is to use randomized trace estimation methods \citep{Hutchinson:1990,MeyerMuscoMusco:2021} to estimate moments of the form $\int_{-1}^1 T_j(x)p(x)dx = \frac{1}{n}\tr(T_j(A))$, where $T_j(x)$ is a degree $j$ polynomial, not equal to $x^j$. If $T_j$ is the $j^\text{th}$ Chebyshev or Legendre polynomial, then it can be shown that only $\poly(1/\epsilon)$ accurate estimates of the first $\ell= c/\e$ moments are needed to approximate the spectral density to $\epsilon$ error in Wasserstein-1 distance \citep{BravermanKrishnanMusco:2022}. A natural question then is, can these general polynomial moments be estimated using random walks in time independent of $n$ for graph adjacency matrices? Unfortunately, it is not known how to do so: the challenge is that the $\ell^\text{th}$ Legendre polynomial or Chebyshev polynomial has coefficients exponentially large in $\ell$, so $\tr(T_j(A))$ cannot be effectively approximated given  a routine for approximating $\tr(A^j)$ for different powers $j$.

\subsection{Our Contributions}
In this paper, we answer \prettyref{ques:question1} negatively. First, we show that exponentially accurate moments are necessary for estimating a  distribution in  Wasserstein-1 distance, even in the special case of distributions that arise as the spectral density of a graph adjacency matrix.
\begin{restatable}{theorem}{thmmomlb} \label{thm:mom_lb}
	For any $\epsilon \in (0, 1/4]$, there exist weighted graphs $G_1$ and $G_2$ (see~\prettyref{def:mom}) with spectral densities $p_1$ and $p_2$, such that:
	\begin{itemize}
		\item The densities are far in Wasserstein-$1$ distance: $W_1(p_1,p_2) \geq \epsilon$.
		\item For \emph{all} positive integers $j$, %
    moments $m_j(p_1) = \int_{-1}^1 x^{j} p_1(x) \rd x$ and $m_j(p_2) = \int_{-1}^1  x^{j}p_2(x) \rd x$ are exponentially close: $(1-\delta)m_j(p_1) \leq m_j(p_2) \leq  (1+\delta)m_j(p_1)$ for some $\delta \leq 16\cdot 2^{-1/4\e}$.
	\end{itemize}
\end{restatable}
\prettyref{thm:mom_lb} shows that \citet{KongValiant:2017}'s requirement that each moment be estimated to accuracy $2^{-O(1/\epsilon)}$ cannot be avoided if we want an $\epsilon$ accurate approximation in Wasserstein distance.
It thus rules out a direct improvement to the analysis of the spectral density estimation algorithm from of \citet{Cohen-SteinerKongSohler:2018}.
 In particular, even if we had a procedure that returned exponentially accurate multiplicative estimates to the moments of a graph's spectral density,\footnote{When run for $O(1/\delta^2)$ steps, the random walk method of \citet{Cohen-SteinerKongSohler:2018} actually achieves a weaker  moment approximation with additive rror $\delta$.
 This is always greater than $\delta m_\ell(p_1)$ because all of $p_1$'s moments are upper bounded by $1$ since it is supported on $[-1,1]$.} and even if it returns such estimates for \emph{all} of the moments (not just the first $O(1/\epsilon)$), then we would not be able to distinguish between $G_1$ and $G_2$.

 Our proof of  \prettyref{thm:mom_lb} is based on a hard instance built using cycle graphs. It is not hard to show that the spectral densities of two disjoint cycles of length $1/\epsilon$ and of one cycle of length $2/\epsilon$ differ by $\epsilon$ in Wasserstein-1 distance. Additionally, it can be shown that the first $c/\epsilon$  moments of these graphs are exponentially close. This example would thus prove \prettyref{thm:mom_lb} if we restricted our attention to moments of degree $j \leq c/\epsilon$.
 However, for the cycle graph, higher moments can be more informative: for example, the $j^\text{th}$ moment for $j=O(1/\epsilon^{2})$ can be shown to distinguish the cycles of different length, even when only estimated to polynomial additive accuracy.
 To see why this is the case, note that, since a random walk of length $O(1/\epsilon^2)$ mixes on the cycle, the probability of it returning in the shorter cycle is roughly twice that as in the longer cycle.

 To avoid this issue, we modify the cycle graph to diminish the value of higher degree moments. In particular, we force all high moments close to zero by creating a graph that consists of many disjoint cycles, either of length $1/\epsilon$ or $2/\epsilon$, joined by a lightweight complete graph on all nodes.
 If weighted correctly, then any walk of length $\Omega(1/\epsilon)$ will exit the cycle it starts in (via the complete graph) with high probability, and the chance of returning to its starting point can be made extremely low by making the graph large enough. At the same time, the lower moments are not effected significantly, so we can show that the graphs remain far in Wasserstein-1 distance.

 \prettyref{thm:mom_lb} has potentially interesting implications beyond showing a limitation for graph spectrum estimation.  For example, related to the discussion about generalized moment methods above, it immediately implies that for any $\ell$, the $\ell^\text{th}$ Chebyshev polynomial cannot be approximated to accuracy $1/\poly(\ell)$ with a polynomial (of any degree!) whose maximum coefficient is $\leq 2^\ell$.
 If it could, we could use less than exponentially accurate measures of the raw moments to approximate the Chebyshev moments, and then use these moments to approximate the spectral density, following \citet{BravermanKrishnanMusco:2022}.  However, by \prettyref{thm:mom_lb}, this is impossible.

While  \prettyref{thm:mom_lb} rules out  direct improvements to the moment-based method of \citet{Cohen-SteinerKongSohler:2018}, it does not rule out the possibility of \emph{some other algorithm} that can estimate the spectral density to $\epsilon$ accuracy using fewer random walk  steps.
For example, we could consider methods that use more information about each random walk than checking whether or not the last step returns to the starting node.
However, our next theorem shows that, in fact, \emph{no such algorithm} can beat the exponential dependence on $1/\epsilon$; we show that, information theoretically, $2^{\Omega(1/\epsilon)}$ samples from random walks started from random nodes are necessary to estimate the spectral density accurately in Wasserstein-1 distance.
\begin{restatable}{theorem}{thmtranscript} \label{thm:transcript_main}
 For any $\e < 1/2$, no algorithm that is given access to the transcript of $m$, length $T$ random walks initiated at $m$ uniformly random nodes in a given graph $G$ can approximate $G$'s spectral density to $\epsilon$ accuracy in the Wasserstein-1 distance with probability $> 3/4$, unless $m\cdot T > \frac{1}{16}\cdot 2^{1/4\epsilon}$.
 \end{restatable}
While more technical, the proof of \prettyref{thm:transcript_main} is based on the same hard instance as \prettyref{thm:mom_lb}. The distribution $\mathcal{D}$ is supported on two graphs that are $\epsilon$ far in Wasserstein distance: a collection of cycles of length $1/\epsilon$ added to a lightweight complete graph, and a collection of cycles of length $2/\epsilon$  added to a lightweight complete graph. We establish that, if node labels are assigned at random, the only way to distinguish between these graphs is to complete a walk around one of the cycles. We show that event happens with exponentially small probability for a random walk of any length.

\subsection{Open Problems and Outlook}
Our main results open a number of interesting directions for future inquiry. Most directly, the bound from \prettyref{thm:transcript_main} is based on an instance involving \emph{weighted graphs}. It would be great to extend the lower bound to unweighted graphs, which are common in practice. While we believe the same lower bound should hold, such an extension is surprisingly tricky: for example, replacing the lightweight complete graph in our hard instances with, e.g., an unweighted expander graph significantly impacts the spectra of both graphs, making them more challenging to analyze.

A bigger open question is to extend our lower bounds to what we call the \emph{adaptive random walk model}, which means that the algorithm is allowed to start a random walk either at a random node, or at any other node it wishes. Since this model allows for e.g. sampling random neighbors of any node, it is closely related to other access models. For example, up to logarithmic factors, the number of random walk steps required in the adaptive model is equal to the number of memory accesses needed when given access to data structure storing an array of neighbors for each node in the graph \cite{BravermanKrishnanMusco:2022}. Currently, the best lower bound we can prove in the adaptive random walk model is that just $\Omega(1/\epsilon^2)$ steps are necessary; we show this result in \prettyref{app:adaptive}. Proving a lower bound exponential in $1/\epsilon$ or finding a faster algorithm that runs in this model would be a nice contribution. Even a conjectured hard instance would be nice -- currently we don't have any.

Finally, we note that our graph-based lower bounds show that, with non-adaptive random walks, it is impossible to distinguish if the spectral densities of two graphs are identical or $\epsilon$-far away in Wasserstein-1 distance with $2^{o(1/\epsilon)}$ steps. Consequently this result constitutes a particular type of hardness for comparing graphs. However, one might consider other notions of graph comparison. For example, in~\prettyref{app:spectrum-comp}, we consider estimating the spectrum of the difference $A_1 - A_2$ between two normalized adjaceny matrices $A_1$ graphs $A_2$ corresponding to graphs $G_1$ and $G_2$ with the same node degress. We show that an $2^{O(1/\epsilon)}$ upper bound is obtainable. Seeking matching upper and lower bounds for this and related problems is another interesting direction for future work.

\subsection{Paper Organization}
In~\prettyref{sec:prelims} we introduce notation and preliminaries. In~\prettyref{sec:mom-lb} we prove a lower bound for spectrum estimation based on moments, establishing ~\prettyref{thm:mom_lb}. In~\prettyref{sec:lb-weighted} we prove  lower bound for spectrum estimation based on random walks, establishing ~\prettyref{thm:transcript_main}. In~\prettyref{app:adaptive}, we give an $\Omega(1/\e^2)$ lower bound for approximating graph spectra in the (stronger) adaptive random walk model. In~\prettyref{app:Lengendre}, we use cycle graphs to construct distributions that are $2/\ell$ far in Wasserstein-1 distance and have the same first $\ell-1$ moments, slightly strengthening a result from~{\citet{KongValiant:2017}}. In~\prettyref{app:spectrum-comp}, we show a new algorithm that uses alternating random walks to estimate the spectrum of the difference of two normalized adjacency matrices.

\section{Preliminaries}\label{sec:prelims}

\noindent \textbf{General notation.} We use $\delta : \R \rightarrow \R$ to denote the indicator function with $\delta(0) \defeq 1$ and $\delta(x) \defeq 0$ for all $x \neq 0$. We use $\mathbf{1}\in\R^{n}$ to denote the all ones vector when $n$ is clear from context. We use $\mathbb{P}[E]$ to denote the probability of an event $E$. We let $E^c$ denote the complement of a random event $E$, so $\mathbb{P}[E^c] = 1- \mathbb{P}[E]$.

\paragraph{Graphs and graph spectra.} We consider undirected graphs $G=(V, E)$ where each edge $e\in E$ has a non-negative weight $w_e \in \R_{\geq 0}$. We call $G$ unweighted when $w_e = 1$ for all $e \in E$. We use $\tilde A\in \R^{V\times V}_{\geq 0}$ to denote the weighted adjacency matrix of $G$ where $\tilde A(v,v') = w_e$ if $e = (v,v')\in E$ and $\tilde A(v,v') = 0$ otherwise. We use $ D \in \R^{V\times V}_{\geq 0}$ to denote the diagonal degree matrix of $G$ where $D$ is diagonal with $D(v,v) \defeq \sum_{e = (v,v') \in E} w_e$ for all $v \in V$. We let $A(G)\in\R^{V\times V}$ denote the normalized adjacency matrix of $G$, i.e.\ $A(G) \defeq D^{-1/2}\tilde A D^{-1/2}$. We refer to $D^{-1}\tilde A$ as the random walk matrix and note that, for degree-regular graphs, $A(G) = D^{-1}\tilde A$.

For an $n$-vertex graph $G$, we let $-1\le\lambda_1\le \lambda_2\le \cdots\le \lambda_{n}\le1$ be the eigenvalues of the normalized adjacency matrix $A(G)$, and use $\blambda = \blambda(G)$ to denote this sorted (in ascending order) eigenvalue list. We let $p(x):[-1,1]\rightarrow [0,1]$ denote the spectral density of $G$, i.e., $p(x) = \frac{1}{n} \sum_{i \in [n]} \delta (x - \lambda_i)$, which is the density of the distribution on $[-1,1]$ induced by $\lambda_i$ (for brevity, we do not distinguish between spectral density and the distribution it induces). We use $m_j(p)$ to denote the $j^{\text{th}}$ moment of $p$, i.e., $m_j(p) = \frac{1}{n}\tr(A(G)^j)$.

\paragraph{Wasserstein distance.}
In this work, we consider the standard Wasserstein-1 distance between distributions, which we may simply refer to as the Wasserstein distance for brevity.
\begin{definition}
The {Wasserstein-$1$ distance} $W_1(p_1, p_2)$ between two
distributions, $p_1$ and $p_2$, supported on the real line
is defined as the minimum cost of moving probability mass in $p_1$ to $p_2$, where the cost of moving probability mass from value $a$ to $b$ is $|a-b|$. Concretely, let $\Psi$ be the set of all couplings  $\psi(x,y)$ between $p_1$ and $p_2$, i.e., $\Psi$ contains all joint distributions $\psi(x,y)$ over $x\in \R$ and $y\in \R$ with marginals equal to $p_1$ and $p_2$. Then:
\begin{align*}
    W_1(p_1, p_2) = \min_{\psi\in \Psi} \int_{\R}\int_{\R}  |x-y|\cdot \psi(x,y)\,\rd x\, \rd y \mper
\end{align*}
\end{definition}
A well known fact is that the Wasserstein-$1$ distance has a dual characterization. Specifically,
\begin{fact}[Kantorovich-Rubinstein Duality \cite{kantorovich1940,kantorovich2006translocation}]
\label{fact:equiv_w1}
\begin{align}\label{def:w1_l1-dual} W_1(p_1,p_2) = \sup_{f: 1-\mathrm{Lipschitz}} \int_{\R} f(x) \cdot (p_1(x) - p_2(x)) \rd x .
\end{align}
\end{fact}
Above, the supremum is taken over all $1$-Lipschitz functions $f$, i.e., that satify $|f(a) - f(b)| \leq |a - b|$ for all $a,b\in \R$.
Overloading notation, for graphs $G_1$ and $G_2$ with spectral densities $p_1$ and $p_2$ respectively, we let $W_1(G_1,G_2) \defeq W_1(p_1,p_2)$ to denote the Wasserstein-$1$ distance between $p_1$  and $p_2$. We note that, for any two $n$-vertex graphs $G_1$ and $G_2$, it can be checked (see, e.g. \cite{KongValiant:2017}) that:
\begin{equation}\label{eq:w1-graph}
W_1(G_1,G_2) = \frac{1}{n} \norm{\blambda(G_1) - \blambda(G_2)}_1 \mper
\end{equation}

\paragraph{Access models.} As discussed in the introduction, we consider several possible data access models  for estimating the spectral density of a normalized graph adjacency matrix, $A(G)$ for $G = (V,E,w)$. First, we consider algorithms that, for some integer $j \geq 0$ and accuracy parameter $\delta$, have access to $\delta$-accurate approximations, $\tilde{m}_1, \ldots, \tilde{m}_j$, to the first $j$ moments of $G$'s spectral density $p$, $m_1(p), \ldots, m_j(p)$. Specifically, we have that $|\tilde{m}_j - m_j(p)| \leq \delta \cdot m_j(p)$.

A natural generalization of the setting where approximate moments are available is to consider algorithms that access $G$ via random walks, since repeated random walks can be used to approximate moments~\cite{Cohen-SteinerKongSohler:2018}. In this work, we primarily consider a \emph{non-adaptive random walk model}, where the algorithm can run $m$ random walks each of length $T \geq 1$, starting at $m$ vertices $v_0^{(1)}, \ldots, v_0^{(m)}$ chosen uniformly at random from $G$. 
For each walk, the algorithm can observe the entire sequence of vertex labels visited in order
We call this information the the walk ``transcript'' and denote the set of transcripts by $S=\{S_1,\cdots, S_m\}$.  Note that, at vertex $v$, the probability that the next vertex in the random walk is equal to $v'$ is the $(v,v')$ entry of $D^{-1} \tilde{A}$.

In~\prettyref{app:adaptive}, we also consider the richer random walk model that we refer to as the \emph{adaptive random walk model} where the algorithm can choose the starting node  $v_0^{(1)}, \ldots, v_0^{(m)}$. This is in contrast to the \emph{non-adaptive random walk model} where starting nodes are uniformly random.

\paragraph{Cycle spectra.} Our lower bound instances in this paper involve collections of cycle graphs. We let $R_c$ denote an undirected cycle graph of length $c$, and we let $R_c^{k}$ denote a collection of $k$ such cycles. Recall that we use $A(R_c^{k})$ to denote the normalized adjacency matrix and $\blambda(R_c^{k})$ for a sorted list of eigenvalues for the normalized adjacency matrix. We leverage the following basic lemma on the spectrum of cycle graphs. 

\begin{lemma}[Eigenvalues of cycle graph]\label{lem:eig_ring}
 For any odd integer $\ell$, the eigenvalues of $A(R_{\ell})$ are $\cos(\frac{2k}{\ell} \pi)$ with multiplicity $2$ for $0 < k < \frac{\ell}{2}$ and $1$ with multiplicity $1$. The eigenvalues of $A(R_{2\ell})$ are $\cos(\frac{k}{\ell}\pi)$ with multiplicity $2$ for $0 < k < \ell$ and $\pm 1$ each with multiplicity $1$.
Further, we have $W_1(R^{2}_{\ell},R_{2\ell}) = 1/\ell$.
\end{lemma}

\begin{proof}
   The eigenvalues of the normalized adjacency matrix of cycle graphs are well known and can be found, e.g., in ~\citet{Spielman:2019}. The Wasserstein distance immediately follows since we have:
\begin{align*}
&  \| \blambda(R_{\ell}^{2})-\blambda(R_{2\ell})\|_1 =  \abs{ 1 - \cos( \pi /\ell) } + \abs{\cos(2 \pi /\ell) - \cos(\pi /\ell)} + \dots + \abs{- 1 - \cos (\pi (\ell-1) / \ell)}\\
  & \hspace{5em}  =  1-\cos( \pi /\ell)+\cos(\pi /\ell)-\cos(2 \pi /\ell)+\cdots+ \cos (\pi (\ell-1) / \ell)-(-1)  = 2\mper    \qedhere
  \end{align*}
\end{proof}

\begin{remark}\label{rmk:jlmomsame}
  The first $j<\ell$ moments of the spectral density of $R^{2}_{\ell}$ and $R_{2\ell}$ are the same. This is true because the number of ways a walk of length $j < \ell$ can return to its starting node is the same in both   $R^{2}_{\ell}$ and $R_{2\ell}$: $2 \cdot \binom{j}{j/2}$ for even $j$ and $0$ for odd $j$.
\end{remark}

\section{Limits on Moment Estimation Methods}\label{sec:mom-lb}

In this section, we construct two weighted graphs $G_1$, $G_2$ with a same number of vertices, i.e., $|V_1| = |V_2|$, that we prove are $\e$-far in Wasserstein distance but have exponentially close moments.
We detail the construction in the definition below.

\begin{definition}\label{def:mom}
$G_1$ is constructed by starting with a collection of $2n\ell$ isolated vertices and $2n$ disjoint cycles, each of size $\ell$. $G_2$ is constructed by starting with a collection $2n\ell$ isolated vertices and $n$ disjoint cycles, each of size $2\ell$. In both graphs, the edges in the cycle have weight $1/4$ and every vertex in a cycle is then connected to all other cycle vertices with weight $1/(4n\ell)$ (including a self-loop); the isolated vertices only have self-loop with weight $1$. We choose $\ell$ to be an odd number and let $n = \lceil 2^\ell/4\rceil$. Note that each graph has $4n\ell$ vertices. 
\end{definition}

\begin{figure}[tb]
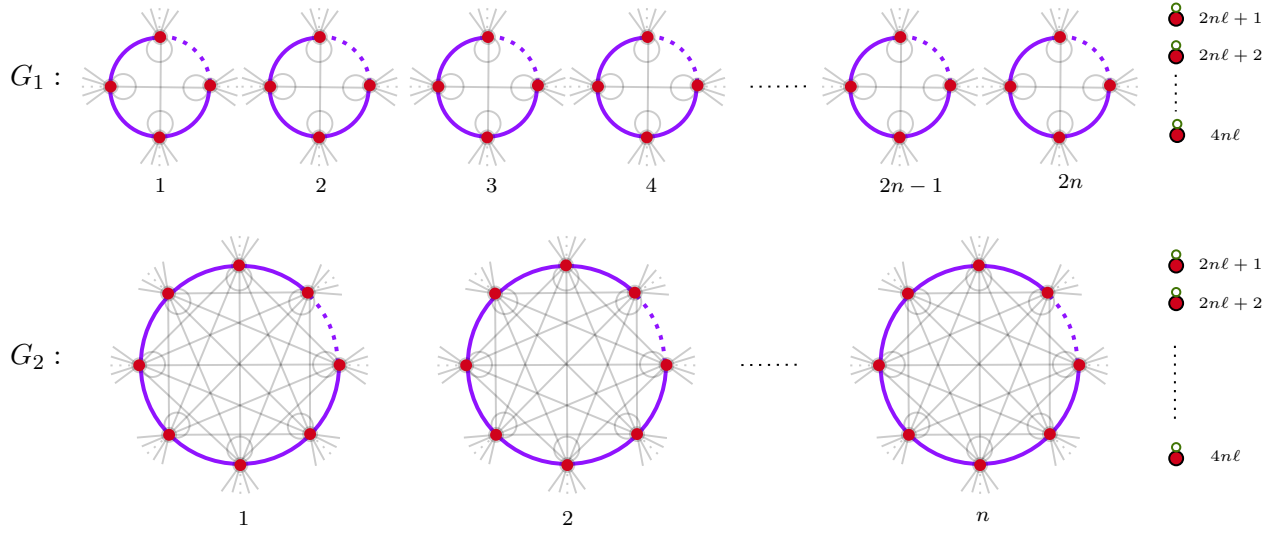

    \centering
\clearpage{}%

\tikzset{every picture/.style={line width=0.75pt}} %

\clearpage{}%

    \vspace{-2em}
    \caption{Diagram depicting graphs $G_1$ and $G_2$ from \prettyref{def:mom}. $G_1$ contains $2n$ cycles of length $\ell$ and $2n\ell$ isolated vertices, and $G_2$ contains $n$ cycles of length $2 \ell$ and $2n\ell$ isolated vertices. In both graphs, the purple edges have weight $1/4 + 1/(4n\ell)$, the grey edges connect all the vertices not connected by purple edges (including self-loops), with weight $1/(4n\ell)$, and the green edges are self-loops of weight $1$. 
    }
    \label{fig:mom_graph}
\end{figure}

See \prettyref{fig:mom_graph} for a visual representation of the construction from \prettyref{def:mom} and \prettyref{fig:mom_plot} for a plot of the spectra of $G_1$ and $G_2$. 
We bound the Wasserstein distance between these spectra below. 

\begin{lemma}\label{lem:w1g1g2}
For weighted graphs $G_1$, $G_2$ constructed in~\prettyref{def:mom}, $W_1(G_1,G_2) = 1/(4\ell)$.
\end{lemma}
\begin{proof}
  Let $\mathbf{I}$ denote a $2n\ell \times 2n\ell$ identity matrix. The normalized adjacency matrices of the two graphs are
  \begin{align*}
  A(G_1) = \begin{bmatrix}
  \frac{1}{2} \cdot A(R_{\ell}^{2n}) + \frac{1}{2} \cdot \frac{1}{2n\ell} \cdot \bm{1} \bm{1}^\top  & 0\\
  0 & \mathbf{I}
  \end{bmatrix},
  \text{ and }
  A(G_2) = \begin{bmatrix}\frac{1}{2} \cdot A(R_{2\ell}^{n}) + \frac{1}{2} \cdot \frac{1}{2n\ell} \cdot \bm{1} \bm{1}^\top & 0 \\
  0 & \mathbf{I}\end{bmatrix}\mper
  \end{align*}
  Recall that we use $R_{\ell}^{2n}$ to denote the graph of $2n$ disjoint cycles of size $\ell$, and $R_{2\ell}^{n}$ to denote the graph
 of $n$ disjoint cycles of size $2\ell$, respectively. 
 Additionally, recall we use $\blambda(G_1)$ and $\blambda(G_2)$ to denote the sorted  (in ascending order) eigenvalues of $A(G_1)$ and $A(G_2)$, and $\blambda(R_\ell^{2n})$ and $\blambda(R_{2\ell}^n)$ for the sorted eigenvalue list of $A(R_{\ell}^{2n})$ and $A(R_{2\ell}^{n})$, respectively.

   Since $A(R_{\ell}^{2n})$ and $A(R_{2\ell}^{n})$  are regular graphs and both commute with $\bm 1\bm 1^\top $, they both share the same eigenvectors with $\bm 1\bm 1^\top $. For simplicity of notation we let $\mathcal{R}_1\defeq R_{\ell}^{2n}$, $\mathcal{R}_2 = R_{2\ell}^{n}$. For $i \in [2]$ we have:
  \begin{equation}\label{eq:eig_g1g2}
    \blambda_j(G_i) = \begin{cases}
    \frac{1}{2}\blambda_j(\mathcal{R}_i) ~&~\text{for}~j \in \{1,2,\cdots,2n\ell-1\}\\
    1~&~\text{for}~j \in \{2n\ell, 2n\ell+1,\cdots, 4n\ell\} \mper
    \end{cases} 
  \end{equation}

  This implies $W_1(G_1,G_2) = \frac{1}{4n\ell} \norm{\blambda(G_1) - \blambda(G_2)}_1 = \frac{1}{2}\cdot\frac{1}{4n\ell} \norm{\blambda(\mathcal{R}_1) - \blambda(\mathcal{R}_2)}_1$ by the characterization of Wasserstein distance given in \prettyref{eq:w1-graph}.
  Thus it suffices to calculate $\norm{\blambda(\mathcal{R}_1) - \blambda(\mathcal{R}_2)}_1$. Since these are disjoint cycles, we only need to focus on the Wasserstein distance between a cycle of size $2\ell$ and $2$ disjoint cycles of size $\ell$. Applying \prettyref{lem:eig_ring}, we get $\norm{\blambda(\mathcal{R}_1) - \blambda(\mathcal{R}_2)}_1 = n\cdot 2\ell \cdot W_1(R_{\ell}^{2},R_{2\ell}) = 2n$.
  Plugging this back  we get the claimed Wasserstein distance $W_1(G_1,G_2) = 1/(4\ell)$.
\end{proof}

\begin{figure}
    \centering
    \includegraphics[width = 0.75\textwidth]{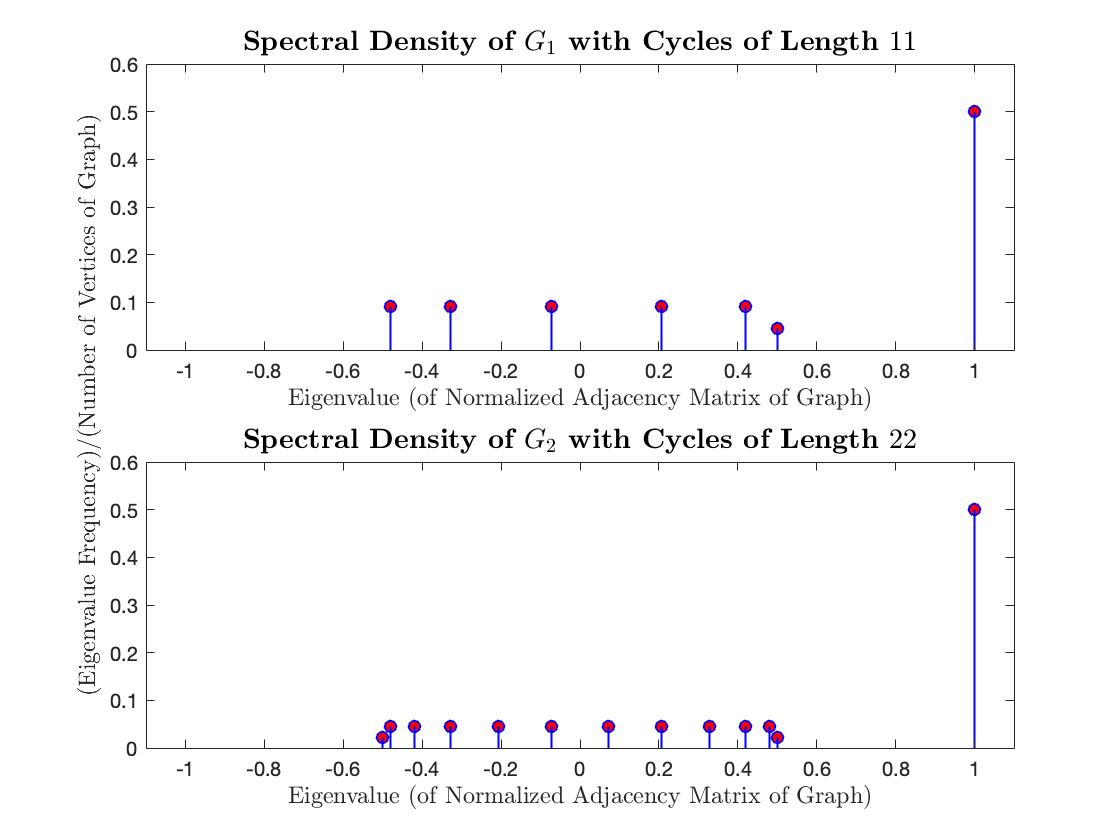}
    \caption{Spectral Density of $G_1$ and $G_2$ as defined in \prettyref{def:mom}, with cycles of length $11$ and $22$, respectively.}
    \label{fig:mom_plot}
\end{figure}

Next we show that the moments of the constructed graphs $G_1$ and $G_2$ are exponentially close.

\begin{lemma}\label{lem:exp_mom_graph}
Let $G_1$ and $G_2$ be weighted graphs as constructed in~\prettyref{def:mom}. Let $p_1, p_2$ be the spectral density of $G_1, G_2$ respectively.  It holds that $m_j(p_i)\in[1/2,1]$ for all $j\ge 0, i=1,2$ and also
\begin{align*}
    \abs{m_j(p_1) - m_j(p_2)} &= 0 \text{ for } j < \ell & &\text{and} & \abs{m_j(p_1) - m_j(p_2)}\leq 2^{-\ell+1} \text{ for } j \geq \ell.
\end{align*}
\end{lemma}

\begin{proof}
For the first claim, we note that $m_j(p_i)\ge \frac{2n\ell}{4n\ell}\cdot 1^j\ge \frac{1}{2}$. The upper bound of $1$ follows trivially given boundedness of all eigenvalues of normalized adjacency matrices.

For $j\ge \ell$, and $i \in [2]$, we also have
\[
m_j(p_i)\le \frac{2n\ell+1}{4n\ell}\cdot 1^j+\frac{2n\ell-1}{4n\ell}\cdot\left(\frac{1}{2}\right)^j\le \frac{1}{2}+\frac{1}{4n\ell}+\frac{1}{2^{j+1}}.
\]
Thus, we can immediately conclude that $m_j(p_i)\in[ \frac{1}{2}, \frac{1}{2}+\frac{1}{2^{j+1}}+\frac{1}{4n\ell}]$ and obtain the claimed bounds for $j\geq \ell$ by plugging in the choice of $n\ge 2^\ell/4$.

For $j<\ell$, we use the fact that $m_j(p_i) = \frac{1}{n} \tr({A(G_i)}^j)$. Using \prettyref{eq:eig_g1g2} we can calculate:
\begin{align}
  \abs{m_j(p_1) - m_j(p_2)} & = \frac{1}{4 n \ell} \cdot \abs{ \sum_{i=1}^{2n\ell-1} \frac{\blambda_i ^j\paren{R^{2n}_{\ell}}}{2^j} +(2 n \ell +1) - \sum_{i=1}^{2n\ell-1} \frac{\blambda_i ^j\paren{R^{n}_{2\ell}}}{2^j} -(2 n \ell +1) } \nonumber \\
  & = \frac{1}{4 n \ell} \cdot \abs{ \sum_{i=1}^{2n\ell} \frac{\blambda_i ^j\paren{R^{2n}_{\ell}} - \blambda_i^j \paren{R^{n}_{2\ell}}}{2^j}  } \mper \label{eq:mjpidiff}
\end{align}
Since $R^{2n}_{\ell}$ and $R^{n}_{2\ell}$ are disjoint cycles, the moments of the spectral density of $R^{2n}_{\ell}$ and $R^{2n}_{\ell}$ are the same as the moments of the spectral density of $R^{2}_{\ell}$ and $R_{2\ell}$.
This is true because the eigenvalues of the disjoint copies of $A(R^{2n}_{\ell})$ and $A(R^{2n}_{\ell})$ are the same as the eigenvalues of the disjoint copies of  $A(R^{2}_{\ell})$ and $A(R_{2\ell})$ with increased multiplicity, which is scaled by the size of the respective graphs.
Since the first $j<\ell$ moments of  $R^{2}_{\ell}$ and $R_{2\ell}$ are the same (see \prettyref{rmk:jlmomsame}), we get from \prettyref{eq:mjpidiff} that
\[ \abs{m_j(p_1) - m_j(p_2)} = \frac{1}{4 n \ell} \cdot \abs{ \sum_{i=1}^{2n\ell} \frac{\blambda_i \paren{R^{2n}_{\ell}}^j - \blambda_i \paren{R^{n}_{2\ell}}^j}{2^j}  } = 0  \mper  \qedhere\]
\end{proof}

We briefly remark that the proof of \prettyref{lem:exp_mom_graph} required picking a value of $n$ that is exponentially large in $\ell$ to ensure that when a random walk leaves the cycle it started from, it only comes back to the same cycle with a very low probability. Otherwise, we would not have been able to show that the higher moments of $G_1$ and $G_2$ ($j \geq \ell$) are close.

\thmmomlb*
  \begin{proof}
The proof of the first statement follows by substituting $\ell$ with the largest odd integer smaller than $1/(4\e)$ in \prettyref{lem:w1g1g2}. Next, we know that for all $j$, $m_j(p_1)\in[1/2,1]$. So, by \prettyref{lem:exp_mom_graph}, $|m_j(p_1)-m_j(p_2)| \le 2^{-\ell+2} m_j(p_1)$. The statement holds since we have $\ell\ge 1/(4\e)-2$.
  \end{proof}
\section{Limits on Random Walk Methods}
\label{sec:lb-weighted}

In this section, we prove \prettyref{thm:transcript_main}, which can be viewed as a strengthening of \prettyref{thm:mom_lb}. While \prettyref{thm:mom_lb} rules out directly improving SDE algorithms like that of \cite{Cohen-SteinerKongSohler:2018} based on estimating moments, \prettyref{thm:transcript_main} shows that no method that performs less than $2^{O(1/\epsilon)}$ steps of non-adaptive random walks in a graph can reliably estimate the spectral density to error $\epsilon$ in Wasserstein distance, whether or not the algorithm is based on moment estimation or not.  

To prove \prettyref{thm:transcript_main}, we construct a hard pair of graphs that are $\epsilon$ far in Wasserstein distance, but difficult to distinguish based on random walks. This pair is identical to the hard instance constructed in the previous section, although without isolated nodes. These nodes were necessary to show that even accurate \emph{relative error} moment estimates do not suffice for spectral density estimation. However, they are not needed for the random-walk lower bound, and eliminating them simplifies the analysis. Formally, the construction is as follows:

\begin{definition}\label{def:weighted_graphs}
$G_1$ is constructed by starting with a collection of $2n$ disjoint cycles, each having size $\ell$ for odd integer $\ell$. The edges in the cycle have weight $1/4$. 
After constructing the cycles, we connect every vertex in $G_1$ to all other vertices with weight $1/(4n\ell)$ (including a self-loop).
$G_2$ is constructed by starting with a collection of $n$ disjoint cycles, each having size $2\ell$. The edges in the cycle have weight $1/4$ and every vertex is then connected to all other vertices with weight $1/(4n\ell)$ (including a self-loop). We choose $n = 2\cdot 2^{2\ell}$.
\end{definition}

\begin{lemma}\label{lem:w1g1g2_rw}
For weighted graphs $G_1$, $G_2$ constructed in~\prettyref{def:weighted_graphs}, $W_1(G_1,G_2) \geq 1/(2\ell)$.
\end{lemma}
\begin{proof}
The normalized adjacency matrices of the two graphs are:
  \begin{align*}
  A(G_1) =
  \frac{1}{2}\cdot A(R_{\ell}^{2n})  +
  \frac{1}{4n\ell}\bm{1} \bm{1}^\top
  \text{ and }
  A(G_2) = \frac{1}{2}\cdot A(R_{2\ell}^{n}) + \frac{1}{4n\ell} \cdot \bm{1} \bm{1}^\top \mper
  \end{align*}
  As before, since $A(R_{\ell}^{2n})$ and $A(R_{2\ell}^{n})$ are degree-regular graphs and both commute with $\bm 1\bm 1^\top $, we can write the sorted vector of eigenvalues $\blambda(G_1),\blambda(G_2)$ of $A(G_1),A(G_2)$ as
  \begin{align*}
  \blambda_j(G_1) & = \begin{cases}
  \frac{1}{2}\blambda_{j}(R_\ell^{2n}) ~&~\text{for}~j\in\{1,\cdots,2n\ell-1\}\\
    1~&~\text{for}~j=2n\ell \mcom
  \end{cases}  \\
  \text{ and }
  \blambda_j(G_2) & = \begin{cases}
  \frac{1}{2}\blambda_j(R_{2\ell}^n) ~&~\text{for}~j\in\{1,\cdots,2n\ell-1\}\\
  1~&~\text{for}~j=2n\ell \mper
  \end{cases} 
  \end{align*}
  Since the top eigenvalues of $R_{2\ell}^n$ and $R_{\ell}^{2n}$ are the same, we conclude that $W_1(G_1,G_2) = \frac{1}{2n\ell} \cdot\left(\frac{1}{2}\cdot\norm{\blambda(R_{\ell}^{2n}) - \blambda(R_{2\ell}^{n})}_1\right)$.
Applying \prettyref{lem:eig_ring}, we have that  $\norm{\blambda(R_\ell^{2n}) - \blambda(R_{2\ell}^n)}_1 = n\cdot 2\ell \cdot W_1(R_{\ell}^{2},R_{2\ell}) = 2n.$ Plugging in, we conclude that
 $W_1(G_1,G_2) = 1/(2\ell)$.
\end{proof}
We next show that the transcripts of randomly started, non-adaptive random walks generated on $G_1$ and $G_2$ have similar distributions.

\begin{definition}\label{def:dg1g2}
    For $m$ non-adaptive random walks, each with length $T$, a random walk transcript $S$ is a collection of $m$ individual walk, $S = \Set{S_1,\dots,S_m}$ where each $S_i$ consists of a list of $T$ node labels $v_{i,1}, \ldots v_{i,T}$ (the nodes visited in the walk).
    Let $\mathcal{D}_{G_1}$ and $\mathcal{D}_{G_2}$ denote probability distribution over random walk transcripts generated when walking in $G_1$ and $G_2$, respectively, with nodes labeled using a uniform random permutation of the integers $1, \ldots, 2n\ell$.
\end{definition}

Our main result is as follows:
\begin{lemma}\label{lem:tv_weighted}
For $m$ non-adaptive walks of length $T$, the total variation distance between  $\mathcal{D}_{G_1}$ and $\mathcal{D}_{G_2}$ (\prettyref{def:dg1g2}) is bounded by $\mathrm{d_{TV}}(\mathcal{D}_{G_1},\mathcal{D}_{G_2}) \leq \frac{2m^2T^2}{n} + \frac{mT}{2^\ell}$.
\end{lemma}

To prove \prettyref{lem:tv_weighted}, we define a coupling between $\mathcal{D}_{G_1}$ and $\mathcal{D}_{G_2}$. We then show that, with high probability, the coupling outputs an identical transcript in both graphs. This establishes closeness in TV distance via the standard coupling lemma, which we state specialized to our setting below:
\begin{fact}
\label{fact:coupling_lemma}
Let $\mathcal{D}$ be any distribution over pairs of random walk transcripts ${S}^1$ and ${S}^2$ such that the marginal distribution of $S^1$ equals $\mathcal{D}_{G_1}$ and the marginal distribution of $S^2$ equals $\mathcal{D}_{G_2}$. Then:
\begin{align*}
    \mathrm{d_{TV}}(\mathcal{D}_{G_1},\mathcal{D}_{G_2}) \leq \mathbb{P}_{\mathcal{D}}[S^1\neq S^2].
\end{align*}
\end{fact}
\begin{proof}[Proof of {\prettyref{lem:tv_weighted}}]
We define a coupling $\mathcal{D}$ by describing a process that explicitly generates two random walk transcripts $S^1$ and $S^2$ which are distributed according to $\mathcal{D}_{G_1}$ and $\mathcal{D}_{G_1}$. To do so, we use a ``lazy labelling'' procedure that randomly labels nodes as they are visited in the random walks. To support that labeling, we define two dictionaries, $L_1: V_1\rightarrow 1, \ldots, 2n\ell$ and $L_2: V_2\rightarrow 1, \ldots, 2n\ell$ that maps the vertex sets of $G_1$ and $G_2$ (denoted as $V_1$ and $V_2$) to labels. Initially, $L_i(v)$ returns NULL for any vertex $v\in V_i$. However, if we set $L_i(v)\leftarrow j$ for a label $j$, then for all future calls to the dictionary, $L_i(v)$ returns $j$. Additionally, in our description of the coupling we will refer to the ``cycle'' that a node $v$ lies in (in $G_1$ or $G_2$) and to $v$'s ``left neighbor'' and ``right neighbor''. Referring to \prettyref{def:weighted_graphs}, these terms refer to the cycle that $v$ would be in if the lightweight copy of the complete graph had not been added to to the graph, and respectively to $v$'s neighbors in that cycle.

With this notation in place, we describe the coupling procedure below.

\begin{enumerate}
\item Choose a random permutation $\Pi$ of the labels $1, \ldots, 2n\ell$. Let $\Pi(j)$ denote the $j^\text{th}$ label in the permutation. Initialize $j \leftarrow 1$.
\item For  $k = 1, \ldots, m$:
\begin{enumerate}
\item Choose independent, uniformly random nodes $v_{k,1}^1$ in $G_1$ and ${v_{k,1}^2}$ in $G_2$. If $L_1(v_{k,1}^1) = \text{NULL}$ (which means that the node has never been visited before in any of our $k-1$ previous random walks) set $L_1(v_{k,1}^1) \leftarrow \Pi(j)$. Likewise, if $L_2(v_{k,1}^2) = \text{NULL}$, set $L_2(v_{k,1}^2) \leftarrow \Pi(j)$. Increment $j \leftarrow j+1$.
\item For $i= 1, \ldots, T$
\begin{enumerate}
\item With probability $\frac{1}{4}$ let $v_{k,i+1}^1$ be the right neighbor of $v_{k,i}^1$ in $G_1$ and let $v_{k,i+1}^2$ be the right neighbor of $v_{k,i}^2$ in $G_2$. With probability $\frac{1}{4}$ let $v_{k,i+1}^1$ be the left neighbor of $v_{k,i}^1$ in $G_1$ and let $v_{k,i+1}^2$ be the left neighbor of $v_{k,i}^2$ in $G_2$. With probability $\frac{1}{2}$, let $v_{k,i+1}^2$ and $v_{k,i+1}^2$  be uniformly random nodes in $G_1$ and $G_2$, respectively. In this last case, which we refer to as the RESET case, $v_{k,i+1}^2$ and $v_{k,i+1}^2$ can be chosen independently.
\item If $L_1(v_{k,i+1}^1) = \text{NULL}$, set $L_1(v_{k,i+1}^1) \leftarrow \Pi(j)$. Likewise, if $L_2(v_{k,i+1}^2) = \text{NULL}$, set $L_2(v_{k,i+1}^1) \leftarrow \Pi(j)$. Increment $j = j+1$.
\end{enumerate}
\end{enumerate}
\item Return \begin{align*}
    S^1 &= \{\{L_1(v_{1,1}^1),\ldots, L_1(v_{1,T}^1)\}, \ldots,\{L_1(v_{m,1}^1),\ldots, L_1(v_{m,T}^1)\}\} \mcom \\
    S^2 &= \{\{L_2(v_{2,1}^1),\ldots, L_2(v_{2,T}^1)\}, \ldots,\{L_2(v_{m,1}^2),\ldots, L_2(v_{m,T}^2)\}\} \mper
\end{align*}
\end{enumerate}
We first observe that the above process is a coupling, as it returns $S^1$ sampled from $\mathcal{D}_{G_1}$ and $S^2$ sampled from $\mathcal{D}_{G_2}$. So, we are left to argue that, with high probability, $S^1 = S^2$.

To do so, we use the fact that the transcripts are identical if two events hold. To define these events, note that each walk in each transcript begins in a cycle in $G_1$ or $G_2$, and then takes a random number of steps left and right in that cycle until ``resetting'' with probability $1/2$ to a uniformly random node in the graph (which could bring the walk to a new cycle, the same cycle it is currently in, or a cycle visited previously). For transcript $S^1$, let $R_1^1, \ldots, R_q^1$ denote the list of cycles visited between each RESET step across all $m$ walks in that transcript. Likewise, let $R_1^2, \ldots, R_q^2$ denote the set of cycles visited in $S^2$. $S^1$ and $S^2$ are always identical if the following events occur:
\begin{description}
    \item[Event 1:] For all $j \neq k$, $R_j^1 \neq R_k^1$ and  $R_j^2  \neq R_k^2$. 
    \item[Event 2:] For all $j\in 1, \ldots, s$, we take fewer than $\ell$ left/right steps in $R_j^1$ and $R_j^2$ before a RESET.
\end{description}
To see why this is the case, note that, if Event 1 occurs, the only way that $S^1$ and $S^2$ would differ is if, while random walking in $R_j^1$ and $R_j^2$, we move to nodes $v^1$ and $v^2$ where $L_1(v^1)$ is defined but $L_2(v^2)$ is NULL, or vice-versa. However, the only way this can happen is if we complete an entire loop around $R_j^1$, and thus return to a node that was previously labeled. Since each $R_j^1$ has $\ell$ nodes, such a loop cannot be completed if we always take $< \ell$ steps before resetting to a new cycle.

We proceed to show that Event 1 and Event 2 both occur with high probability. First, consider Event 1. We take at most $mT$ RESET steps across all $m$ random walks. At each step, the probability we return to a cycle we had previously visited is at most $\frac{mT}{2n}$ for the walk in $G_1$ and at most $\frac{mT}{n}$ for the walk in $G_2$. So, by a union bound, we do not return to any previously visited cycle after a RESET in either walk with probability:
\begin{align}
\label{eq:event1_bound}
\Pr{\text{Event 1}} \geq 1 - mT\cdot \frac{mT}{2n} - mT\cdot \frac{mT}{n} \geq  1 -  \frac{2m^2T^2}{n}.
\end{align}
Next, consider Event 2. Note that, since we take a left/right step with probability $1/2$ (and RESET with probability $1/2$) the chance that we take $\ell$ steps or more in a given cycle is equal to $(1/2)^\ell$. We visit at most $q = mT$ cycles, so by a union bound, we take less then $\ell$ left/right steps in each cycle with probability:
\begin{align}
\label{eq:event2_bound}
\Pr{\text{Event 2}} \geq 1 - \frac{mT}{2^\ell}.
\end{align}
Combining \eqref{eq:event1_bound} and \eqref{eq:event2_bound} with a union bound, we conclude that both events hold, and thus $S^1 = S^2$ with probability at least $1 - \frac{2m^2T^2}{n} - \frac{mT}{2^\ell}$. Combined with \prettyref{fact:coupling_lemma}, this proves the lemma. 
\end{proof}

With \prettyref{lem:tv_weighted} in place, we can prove our main lower bound result for non-adaptive random walks: 
\thmtranscript*

\begin{proof}
The theorem follows from \prettyref{lem:w1g1g2_rw} and \prettyref{lem:tv_weighted}. In particular, choose $\ell$ to equal the largest odd integer smaller than $1/4\epsilon$ and choose $n = 2\cdot 2^{2\ell}$. Then consider graphs $G_1$ and $G_2$ generated as in \prettyref{def:weighted_graphs} with random node labels. By \prettyref{lem:w1g1g2_rw}, $W_1(G_1,G_2)> 2\e$. So, there is no distribution $p$ that is $\epsilon$-close in Wasserstein distance to both the spectral density of $G_1$ and $G_2$. Accordingly, any algorithm that estimates the SDE of a graph to error $\epsilon$ with probability $3/4$ can be used to correctly distinguish samples from $\mathcal{D}_{G_1}$ and $\mathcal{D}_{G_2}$ with probability $3/4$. However, with $n$ set as above, we have that $\mathrm{d_{TV}}(\mathcal{D}_{G_1},\mathcal{D}_{G_2}) \leq \frac{2m^2T^2}{n} + \frac{mT}{2^\ell} = \frac{m^2T^2}{2^{2\ell}} + \frac{mT}{2^\ell}$. And thus we can check that $\mathrm{d_{TV}}(\mathcal{D}_{G_1},\mathcal{D}_{G_2}) \leq 1/2$ whenever $mT \leq \frac{1}{4}2^{1/4\epsilon -2}$. As is standard, no algorithm can distinguish between samples from two distributions with TV distance $\delta$ with probability greater than $\frac{1}{2} + \frac{1}{2}\delta$, which establishes the result: any method that correctly distinguishes $\mathcal{D}_{G_1}$ and $\mathcal{D}_{G_2}$ with probability $>3/4$ must use $mT > \frac{1}{4}2^{1/4\epsilon -2}$ random walk steps.
\end{proof}

\section*{Acknowledgements}

We would like to thank Aditya Krishnan for early discussions on the questions addressed in this paper.
Aaron Sidford was supported by a Microsoft Research Faculty Fellowship, NSF CAREER Award CCF-1844855, NSF Grant CCF-1955039, a PayPal research award, and a Sloan Research Fellowship. This work was also supported by NSF Award CCF-2045590.

\appendix

\printbibliography

\section{Lower Bound for the Adaptive Random Walk Model}
\label{app:adaptive}
In this section, we consider lower bounds against a possibly richer class of spectral density estimation algorithms that can access graphs via \emph{adaptive} random walks. Specifically, the algorithm is allowed to start random walks (of any length) at any node of its choosing and can store the entire transcript of these walks. In the adaptive model, the algorithm also has the ability to uniformly sample nodes from the graph, as in the non-adaptive random walk model considered for \prettyref{thm:transcript_main}. 

Interestingly, an adaptive algorithm can solve the hard instance from  \prettyref{thm:transcript_main} using roughly ${O}(\log(1/\epsilon)/\epsilon)$ random walk steps. Specifically, for any node, the algorithm can identify its adjacent cycle nodes with high probability by taking a logarithmic number of $1$-step random walks and identifying the two nodes that are visited most frequently. This allows it to walk one way around the cycle, check its length, and thus distinguish between $G_1$ and $G_2$.

Proving a lower bound in the adaptive random walk setting appears to be much harder than the non-adaptive setting, and we do not have any proposed constructions that we conjecture could establish that $2^{O(1/\epsilon)}$ random walk steps are necessary. However, in this section we give a simple argument for a lower bound of $\Omega(1/\epsilon^2)$ steps. The lower bound is via a reduction to a natural sampling problem, introduced below.

\begin{problem}
\label{prob:marbles}
For a parameter $\alpha \in (1/2,1)$ and integer $n$, suppose we have a jar that contains either $\alpha\cdot n$ red marbles and $(1-\alpha)\cdot n$ blue marbles (Case 1) or contains $(1-\alpha)\cdot n$ red marbles and $\alpha \cdot n$ blue marbles (Case 2). Our goal is to determine if we are in Case 1 or 2 given a sample of $s$ marbles drawn without replacement from the jar. 
\end{problem}

\begin{lemma}
\label{lem:marbles}
Let $\epsilon \in (0,1/2)$, let $\alpha = (1+\epsilon)/2$, and let $n = 2/\epsilon^4$. There is no algorithm that solves \prettyref{prob:marbles} with probability $> 3/4$ unless $s > 1/(4\epsilon^2)$.
\end{lemma}
\begin{proof}
Suppose we draw $s$ marbles from the jar and encode the result in a length $s$ vector (e.g., with a $0$ at position $i$ if the $i^\text{th}$ marble drawn is red, and a $1$ if it is blue). Let $X_1^{(s)}$ denote the distribution over vectors observed in Case 1, and let $X_2^{(s)}$ denote the distribution for Case 2. We will show that $\mathrm{d_{TV}}(X_1^{(s)}, {X}_2^{(s)})$ is small. To do so, we introduce two auxiliary distributions: let $\hat{X}_1^{(s)}$ denote the distribution over vectors observed if we are in Case 1 and draw marbles randomly \emph{with replacement} and let $\hat{X}_2^{(s)}$ denote the distribution if we are in Case 2 and draw marbles \emph{with replacement}.

We first show that $\mathrm{d_{TV}}(X_i^{(s)}, \hat{X}_i^{(s)})$ is small for $i \in \{1,2\}$ when $n$ is large. To do so, let $\cE$ be the event that in $s$ independent draws with replacement, we never pick a previously picked marble. Let $[\hat{X}_i^{(s)}]_{\cE}$ denote the distribution $\hat{X}_i^{(s)}$ conditioned on $\cE$, and note that $[\hat{X}_i^{(s)}]_{\cE}  = {X}_i^{(s)}$. The probability that $\cE$ happens is equal to $1\cdot (1-\frac{1}{n})\cdot(1-\frac{2}{n}) \cdot (1-\frac{s}{n}) \geq 1 - \frac{s^2}{n}$.  Therefore, we conclude that:
\begin{align}
\label{eq:dtv1}
    \mathrm{d_{TV}}(\hat{X}_i^{(m)},{X}_i^{(m)}) \leq  s^2/n.
\end{align}

Next, we show that  $\mathrm{d_{TV}}(\hat{X}_1^{(s)}, \hat{X}_2^{(s)})$ is small. Doing so is equivalent to bounding the total variation distance between $s$ independent draws from a Bernoulli distribution with mean $1-\alpha$ and $s$ independent draws from Bernoulli distribution with mean  $\alpha$. Let $D_\mathrm{KL}(p,q)$ denote the Kullback–Leibler divergence between distributions $p$ and $q$. Applying Pinsker's inequality, we have:
  \begin{align}
  \label{eq:dtv2}
    \mathrm{d_{TV}}(\hat{X}_1^{(s)},\hat{X}_2^{(s)} ) & \le \sqrt{\frac{1}{2}}\sqrt{D_\mathrm{KL}(\hat{X}_1^{(s)},\hat{X}_2^{(s)})}= \sqrt{\frac{s}{2}}\sqrt{D_\mathrm{KL}(\ber(1-\alpha),\ber(\alpha))} \nonumber\\
    & = \sqrt{\frac{s}{2}}\sqrt{\alpha\log(\alpha/(1-\alpha))+(1-\alpha)\log((1-\alpha)/\alpha)} \nonumber\\
        &\leq \sqrt{\frac{s}{2}}\cdot \epsilon.
    \end{align}
    The last inequality holds for any $\alpha$ equal to $(1+\epsilon)/2$ whenever $\epsilon \leq 1/2.$
Applying triangle inequality to combine \eqref{eq:dtv1} and \eqref{eq:dtv2}, we have that:
\begin{align*}
     \mathrm{d_{TV}}(X_1^{(s)},X_2^{(s)}) & \leq  \mathrm{d_{TV}}(\hat{X}_1^{(s)},{X}_1^{(s)}) +  \mathrm{d_{TV}}(\hat{X}_2^{(s)},{X}_2^{(s)}) +  \mathrm{d_{TV}}(\hat{X}_1^{(s)},\hat{X}_2^{(s)})  \leq \frac{2s^2}{n} + \sqrt{\frac{s}{2}}\cdot \epsilon.
\end{align*}
For any $s \leq 1/(4\epsilon^2)$ and $n = 2/\epsilon^4$ we conclude that $\mathrm{d_{TV}}(X_1^{(m)},X_2^{(m)}) < 1/2$. Accordingly, no algorithm can distinguish between $X_1^{(s)}$ and $X_2^{(s)}$ with probability $\geq 3/4$ unless $s > 1/(4\epsilon^2)$.
\end{proof}

With \prettyref{lem:marbles} in place, we are now ready to prove our lower bound for spectral density estimation. To do so, we will show that any adaptive random walk algorithm that can estimate the spectral density of a graph to accuracy $\epsilon$ using $s$ total random walk steps can solve the \prettyref{prob:marbles} using $\leq s$ samples. This reduction requires introducing a second pair of ``hard graphs'' that are close in Wasserstein distance.
In comparison to the hard instance in \prettyref{thm:transcript_main}, these graphs are also based on collection of cycles. The main difference is that we consider two graphs that each contain a mixture of cycles of length  $2\ell$ and $\ell$, but in different proportions.  

\begin{definition}\label{def:adaptive_g1g2}
 For odd integer $\ell$ and parameter $\alpha\in(0.5,1)$, let $G_1$ be a collection of $\alpha n$ disjoint cycles of length $2\ell$ and $2(1-\alpha)n$ cycles of size $\ell$. Similarly, let $G_2$ be a collection of $(1-\alpha)n$ cycles of length $2\ell$ and $2\alpha n$ cycles of size $\ell$. Both graphs have $2n\ell$ vertices in total.
\end{definition}

We use the following expression for the Wasserstein distance between the spectra of the two graphs.

\begin{lemma}\label{lem:dist_mix}
  Let $G_1$ and $G_2$ be unweighted graphs as in~\prettyref{def:adaptive_g1g2}. $W_1(G_1,G_2) = \frac{(2\alpha-1)}{\ell}$.
\end{lemma}
\begin{proof}
  We can compute the exact eigenvalues of the two graphs by combining \prettyref{lem:eig_ring} with the fact that eigenvalues just increase in multiplicity with repeated components. As in that lemma, recall we use $R_{\ell}$ to denote a cycle of length $\ell$ and $R_{2\ell}$ to denote a cycle of length $2\ell$. The Wasserstein distance between  $R_{\ell}^{2}$ and $R_{2\ell}$ is $W_1(R_{\ell}^{2},R_{2\ell}) = 1/\ell$. 
  
  Note that $G_1$ and $G_2$ both have $(1-\alpha)n$ cycles of length $2\ell$ and $2(1-\alpha) n$ cycles of length $\ell$, while $G_1$ has $(2\alpha-1)n$ extra $R_{2\ell}$ cycles and $G_2$ has $(2\alpha-1)n$ extra copies of $R_{\ell}^{2}$. Let $p_1(x)$ and $p_2(x)$ be the spectral density of $G_1$ and $G_2$ respectively, and let $\tilde{p}_1(x)$ and $\tilde{p}_2(x)$ be the spectral density of $R_{2\ell}^{(2\alpha-1)n}$ and $R_{\ell}^{2(2\alpha-1)n}$, respectively. We have that $2n\ell\cdot (p_1(x)-p_2(x)) = (2(2\alpha-1)n\ell \cdot(\tilde{p}_1(x) - \tilde{p}_2(x))$ for all $x\in[-1,1]$. Thus, due to the dual characterization of Wasserstein distance in~\eqref{def:w1_l1-dual}, 
  \[W_1(G_1,G_2) = (2\alpha -1)\cdot W_1(R_{2\ell}^{(2\alpha-1)n},R_{\ell}^{2(2\alpha-1)n}) = (2\alpha -1)\cdot W_1(R_{2\ell},R_{\ell}^{2})= \frac{(2\alpha-1)}{\ell}. \qedhere \]
\end{proof}

We now have all the ingredients in place to prove the main result of this section:
\begin{theorem} \label{thm:wt_graph}
For any $\epsilon < 1/6$, no algorithm that takes $s$ \emph{adaptive} random walks steps in a given graph $G$ can approximate $G$'s spectral density to $\epsilon$ accuracy in the Wasserstein-1 distance with probability $> 3/4$, unless $s \geq 1/(36\epsilon^2)$. 
\end{theorem}
\begin{proof}
Suppose we had such an algorithm (call it $\mathcal{A}$) that uses $s <1/(4\epsilon^2)$ random walk steps to output an $\epsilon/3$-accurate spectral density with probability greater than $3/4$. We will show that the algorithm could be used to solve \prettyref{prob:marbles} using $<1/(4\epsilon^2)$ samples from the jar with probability greater than $3/4$, which is impossible by \prettyref{lem:marbles}. 

To prove this reduction we associated an instance of \prettyref{prob:marbles} with a hidden graph $G$ that is either isomorphic to $G_1$ or $G_2$ as defined in \prettyref{def:adaptive_g1g2}. To make the association, every marble will correspond to $2\ell$ vertices in the graph with some fixed set of known labels. However, the connections between those nodes is hidden. In particular, if the marble is red, the $2\ell$ vertices are arranged in a single cycle of length $2\ell$. Otherwise, they are arranged in two cycles of length $\ell$. The ordering of nodes in both cases is known in advance, but we do not know which of the two cases we are in. Also note that there are no other connections between vertices. Observe that if we are in Case 1 for \prettyref{prob:marbles}, $G$ is isomorphic to $G_1$ and if we are in Case 2, $G$ is isomorphic to $G_2$. So in particular, $G$'s spectral density is either equal to the spectral density of $G_1$ or $G_2$. 

Our main claim is that we can run algorithm $\mathcal{A}$ on the hidden graph $G$ while only accessing $s$ marbles from the jar. To so do, every time the algorithm requests to visit a specific node, we draw the marble from the jar associated with that node's label. In doing so, we learned all edges in the ring containing that node (as well as other edges), so we can perform any future random walk steps initiated from that node. Since $\mathcal{A}$ takes $s$ steps, we at most need to draw $s$ marbles over the course of running the algorithm. At the same time, note that when we choose $\ell = 1$ (considering self loops) and $\alpha = (1+\epsilon)/2$ as in \prettyref{lem:marbles}, \prettyref{lem:dist_mix} implies that the Wasserstein distance between $G_1$ and $G_2$ is equal to $\epsilon$. So, if $\mathcal{A}$ returns an $\epsilon/3$-accurate spectral density with probability $3/4$, we can determine if we are in Case 1 or Case 2 with probability $3/4$, violating \prettyref{lem:marbles}. We conclude that no such algorithm can exist. 

The final statement of the theorem follows by adjusting constants on $\epsilon$.
\end{proof}
\section{Wasserstein Distance Bounds via Chebyshev Polynomials}\label{app:Lengendre}

In this section, we give an alternative proof of a lower-bound by \citet{KongValiant:2017}, which shows that there exist distributions whose first $\ell-1$ moments match exactly, but the Wasserstein distance between the distributions is greater than $1/(2(\ell+1))$. Our analysis tightens their result by a factor of $\sim 4$, showing two such distributions with Wasserstein distance  $2/\ell$. Moreover, we prove that the Wasserstein distance is $\Omega(\ell^{-1})$ for \emph{any} distributions $p$, $q$ whose first $\ell-1$ moments are the same and whose  $\ell$-th moments differ by $\Omega(2^{-\ell})$.

\begin{lemma}[Improvement of {\cite[Proposition 2]{KongValiant:2017}}] \label{prop:kvprop}
	For any odd $\ell$, there exists a pair of distributions $p$, $q$, each consisting of $(\ell+1)/2$ point masses, supported within the unit interval $[-1,1]$ such that $p$ and $q$ have identical first $\ell-1$ moments, and the Wasserstein distance $W_1(p,q) \geq 2/\ell$.
\end{lemma}
\begin{proof}
Recall that we use $R_{\ell}^{2}$ to denote $2$ disjoint cycles of length $\ell$, and use $R_{2\ell}$ to denote a cycle of length $2\ell$, where $\ell$ is an odd number.
	We know the spectrum of $R_{\ell}^{2}$ and $R_{2\ell}$ from \prettyref{lem:eig_ring}. Let $p'$ and $q'$ denote the spectral density of $A(R_{\ell}^{2})$ and $A(R_{2\ell})$.
	We first note that the first $\ell-1$ moments of the spectral density of $p'$ and $q'$ are the same because a random walk of length $\ell-1$ cannot distinguish $R_{\ell}^{2}$ from $R_{2\ell}$ (see \prettyref{rmk:jlmomsame}).

  Also, recall we use $\blambda(R_{\ell}^2)$ to denote the sorted eigenvalue list of $A(R_{\ell}^{2})$ and $\blambda(R_{2\ell})$ to denote the sorted eigenvalue list of $A(R_{2\ell})$. We make the following observations about the spectrum of $A(R_{\ell}^{2})$ and $A(R_{2\ell})$ based on \prettyref{lem:eig_ring}.
	\begin{enumerate}
		\item $A(R_{\ell}^{2})$ has $(\ell+1)/2$ unique eigenvalues, and $A(R_{2\ell})$ has $\ell+1$ unique eigenvalues.
		\item All eigenvalues of $A(R_{\ell}^{2})$ overlap with eigenvalues of $A(R_{2\ell})$. In particular, all the eigenvalues of  $A(R_{\ell}^{2})$ occur two times more in frequency than the corresponding eigenvalues in $A(R_{2\ell})$. Formally, $\forall \lambda \in \blambda(R_{\ell}^2)$,
		\begin{equation}\label{eq:eig_same}
			 \abs{ \Set{ j ~:~ \lambda_j \in \blambda(R_{\ell}^2),  \lambda_j = \lambda, j \in [2\ell] } } = 2 \cdot \abs{ \Set{ j ~:~ \lambda_j \in \blambda(R_{2\ell}), \lambda_j = \lambda, j \in [2\ell] } } \mper
		\end{equation}
		\item All the eigenvalues of $A(R_{2\ell})$ lies in $[-1,1]$.
	\end{enumerate}
	Let $\bLambda^{(2)}$ denote the sorted list of eigenvalues where we remove all the eigenvalues from $\blambda(R_{2\ell})$ that occurs in $\blambda(R_{\ell}^2)$. Let $\bLambda^{(1)}$ be the set of removed eigenvalues. The following observations follow from \prettyref{eq:eig_same}. The size of $\bLambda^{(2)}$, and  $\bLambda^{(1)}$  is $\ell$.
	Moreover $\bLambda^{(1)}$ has the same eigenvalues as $\blambda(R_{\ell}^2)$ where the frequency of each unique eigenvalue is $\blambda(R_{\ell}^2)$ is reduced by a factor of $2$.
    Consequently, we define $p(x) = \frac{1}{\ell} \sum_{j \in [\ell]} \delta\paren{x - \bLambda_j^{(1)}} = p'(x)$, and $q(x) = \frac{1}{\ell} \sum_{j \in [\ell]} \delta\paren{x - \bLambda_j^{(2)}} = 2q'(x) -  p'(x)$. This ensures that $p, q$ are valid distributions and have a support size of $(\ell+1)/2$.

	Since $p'$, and $q'$ have the same first $\ell-1$ moments, we have $p$ and $q$ also have the same first $\ell-1$ moments.
	Moreover, $W_1(p,q) =  W_1(2q' - p',p') = 2 W_1(q',p') = \frac{2}{\ell}$, where the penultimate equality follows from the dual characterization of Wasserstein distance in~\eqref{def:w1_l1-dual} and the last equality follows from \prettyref{lem:eig_ring}.
\end{proof}

We complement \prettyref{prop:kvprop} with the following \prettyref{lem:leg_w1}, which shows that for two distributions $p$ and $q$ such that all their first $\ell-1$ moments are the same and the $\ell$-th moment differ only by $\Omega(2^{-\ell})$, even then the Wasserstein distance between $p$, $q$ is large. The proof follows just by using the fact that there are $1$-Lipschitz polynomials with large leading coefficient.
We note the following standard facts about the Chebyshev polynomials which can, for example, be found in \cite{Chebyshev}.

\begin{fact}\label{fact:cheb}
    The Chebyshev polynomials of the first kind of degree $i, (i\in \N)$, denoted by $T_i(x)$,  satisfy the following properties:
    \begin{enumerate}
        \item $\forall i \in \N$, $\forall x \in [-1,1]$, $\Abs{T_i(x)} \leq 1$.
        \item The leading coefficient of $T_i$ is $2^{i-1}$.
    \end{enumerate}
\end{fact}

\begin{lemma} \label{lem:leg_w1}
    Consider two distributions $p$ and $q$ supported on $[-1,1]$ such that the difference of their first $\ell-1$ moments are $0$ and the difference of their $\ell$-th moment is $c \cdot 2^{-\ell}$. Then, for such a distribution, their Wasserstein distance is bounded by
    \begin{equation*}
        W_1(p,q) \geq \frac{c}{4 \ell} \mper
    \end{equation*}
\end{lemma}

\begin{proof}
    We use the dual characterization of the Wasserstein distance in~\prettyref{def:w1_l1-dual} and consequently, it suffices to exhibit a $1$-Lipschitz function $g$ which has a high inner-product with $p-q$. Let $T_{\ell-1}$ be a degree $\ell-1$ Chebyshev polynomial.
		From \prettyref{fact:cheb} we know that $f_{\ell}(x) = \int T_{\ell-1}(x) \rd x $ is a degree $\ell$, $1$-Lipschitz polynomial in $[-1,1]$, with leading coefficient  $2^{\ell-2}/\ell$. Define $g_{\ell}(x)$ as follows: 
		\begin{equation*}
			g_{\ell}(x) \defeq \begin{cases}
			f_{\ell}(-1), & \text{for } x \in (-\infty, -1) \\
			f_{\ell}(x), & \text{for } x \in [-1,1] \\
			f_{\ell}(1), & \text{for } x \in (1, \infty)\mper 
		\end{cases} 
		\end{equation*}
		From properties of $f_{\ell}(x)$ and by construction, we know that $g_\ell(x)$ is a $1$-Lipschitz function. Therefore,
    \begin{align*}
        W_1(p,q) & \geq \abs{\int_{\R} g_{\ell}(x) (p(x) - q(x)) \rd x} = \abs{\int_{-1}^{1} f_{\ell}(x) (p(x) - q(x)) \rd x} \\
				& = \abs{\int_{-1}^{1} \frac{2^{\ell-2}}{\ell}  x^\ell (p(x) - q(x)) \rd x}
         = \frac{1}{4 \ell} ~ 2^\ell \cdot c 2^{-\ell} = {c} \cdot \frac{1}{4 \ell} \mcom
    \end{align*}
    where the first equality holds because $p(x)-q(x) = 0$ outside $[-1,1]$ and the second equality follows from the fact that the difference of the first $1,\dots,\ell-1$ moments are $0$.
\end{proof}
\section{Another Spectral Metric for Graph Comparison}\label{app:spectrum-comp}

Throughout this section we consider two graphs $G_1$, $G_2$ with the same vertex size $n$ and same vertex labeling $V = [n]$, and their un-normalized adjacency matrix $\tilde A_1$ and $\tilde A_2$ with a common degree matrix $D$. Here we consider learning the spectrum of their difference matrix, i.e., $A(G_1)-A(G_2) = D^{-1/2}\tilde A_1D^{-1/2}-D^{-1/2}\tilde A_2D^{-1/2}$, or equivalently $D^{-1}(\tilde A_1-\tilde A_2)$. We provide a simple proof that $\exp(O(1/\e))$ number of samples also suffice to estimate this distribution up to $\e$-Wasserstein distance, using similar techniques as in~\citet{Cohen-SteinerKongSohler:2018}.

We first restate the main theorem in~\citet{KongValiant:2017} for completeness.

\begin{theorem}[{\citet[Proposition 1]{KongValiant:2017}}]\label{thm:cite-main}
Given two distributions with respective density functions $p,q$ supported on $[a,b]$ whose first $k$ moments are $\alpha = (\alpha_1,\cdots, \alpha_k)$ and $\beta = (\beta_1,\cdots, \beta_k)$, respectively. The Wasserstein distance $W_1(p,q)$ between $p,q$ is bounded by $W_1(p,q)\le C(\frac{b-a}{k}+3^k(b-a)\|\alpha-\beta\|_2)$ for some absolute constant $C$.
\end{theorem}

We define a variant of the non-adaptive random walk access model, represented via an oracle $\tOnaran(G_1,G_2,j,\{x_i\}_{i\in[j]})$, specifically for this problem, which outputs the random trajectory after taking a length $j$ random walk starting from a uniformly randomly chosen vertex, where at step $i\in[j]$ it the follows probabilistic transition of $D^{-1}\tilde A_1$ when $x_i=1$ and $D^{-1}\tilde A_2$ when $x_i = 0$. We consider~\prettyref{alg:spectrum-comp} for estimating the spectral density of matrix $D^{-1}(\tilde A_1-\tilde A_2)$.

\prettyref{alg:spectrum-comp} computes estimates of the moments of difference matrix $D^{-1}(\tilde A_1-\tilde A_2)$. Together with the procedure of computing a distribution based on first $k$ moments using linear programming as stated in~\citet{Cohen-SteinerKongSohler:2018}, we have the following guarantee.

\begin{theorem}\label{thm:side-comp}
Given any two graphs $G_1$, $G_2$ on same set of vertices with a common degree matrix $D$,~\prettyref{alg:spectrum-comp} with $k = 4C/\e$ and $\theta = \e/(3^{2k+2})$ outputs a distribution $p$ that is $\e$-close in Wasserstein-$1$ distance with the spectral density function of $A(G_1)-A(G_2)$ with probability $0.9$, using a total of $2^{O(1/\e)}$ calls to $\tOnaran(G_1,G_2,j,\cdot)$, $j\in[O(1/\e)]$.
\end{theorem}

\begin{algorithm}[t!]
\caption{Spectral Density of Difference of Adjacency Matrices}\label{alg:spectrum-comp}
\DontPrintSemicolon
\textbf{Input:} Graphs $G_1$, $G_2$, oracle $\tOnaran$, accuracy $\e$, probability $\delta$\;
\textbf{Parameters:} $k\in\mathbb{Z}_{+}$, $\theta>0$\;
\For{$j\in[k]$}{
Initialize $\hat{p}_{j}=0$\;
    \For{$(x_1,x_2,\dots, x_j)\in\{0,1\}^j$}{
        Generate $\frac{1}{2}\theta^{-2} j 4^j\log(2k/\delta)$ independent samples of $\tOnaran(G_1,G_2,j, \{x_j\}_{j\in[k]})$\;\label{line:NARAN}
        Let $\hat{p}_{j,x}$ be the fraction of the trajectories which start and end at the same vertex\;\label{line:p}
        Update $\hat{p}_{j} \gets \hat{p}_{j}+\hat{p}_{j,x}$\;
}
}
Construct a distribution $p$ on $[-1,1]$ with first $k$ moments equal to $\{\hat{p}_j\}_{j\in[k]}$\;
\textbf{Return:} $p$
\end{algorithm}

\begin{proof}
Note similarity transformation doesn't affect eigenvalues, thus it suffices to estimate the spectral density function of matrix $D^{-1}(\tilde A_1-\tilde A_2)$, whose $j^\text{th}$ moment is $\frac{1}{n}\mathrm{tr}((D^{-1}\tilde A_1-D^{-1}\tilde A_2)^j)$.

For any $j\in\mathbb{Z}_{+}$, note that
\[\frac{1}{n}\mathrm{tr}((D^{-1}\tilde A_1-D^{-1}\tilde A_2)^j)= \sum_{ x_1,x_2,\cdots, x_j\in\{0,1\}} \frac{1}{n}\mathrm{tr}\left(\prod_{i=1,\cdots, j}(x_i\cdot D^{-1}\tilde A_1+(1-x_i)\cdot D^{-1}\tilde A_2)\right).\]

Given any $x = (x_1,\cdots, x_j)$, we run an alternating random walk as in~$\tOnaran$ to generate unbiased samples of term $\frac{1}{n}\mathrm{tr}(\prod_{i=1,\cdots, j}(x_i\cdot D^{-1}\tilde A_1+(1-x_i)\cdot D^{-1}\tilde A_2))$ (as in~Line~\ref{line:NARAN}). By concentration we can estimate each term $\frac{1}{n}\mathrm{tr}(\prod_{i=1,\cdots, j}(x_i\cdot D^{-1}\tilde A_1+(1-x_i)\cdot D^{-1}\tilde A_2))$ using $\hat{p}_{j,x}$ (as in~Line~\ref{line:p}) up to $\theta/2^j$ additive accuracy with high probability $1-\delta/(k2^j)$ using a total of $\frac{1}{2}\theta^{-2} j 4^j\log(2k/\delta)$ calls to $\tOnaran(G_1,G_2,j,\{x_i\}_{i\in[j]})$. Consequently, using a union bound we have with probability $1-\delta$, $\hat{p}_j$ estimates the $j^\text{th}$ moments up to $\theta$ additive accuracy, each using a total of $O(\theta^{-2} j 2^{3j}\log(2k/\delta))$ calls to some $\tOnaran(G_1,G_2,j,\cdot)$ for all $j\in[k]$.

Picking $k = \frac{4C}{\e}$, $\theta =\frac{\e}{3^{2k+2}}$, we can apply~\prettyref{thm:cite-main} to conclude that the constructed distribution $p$ is an $\e$-approximation in Wasserstein distance to the spectral density function of $A(G_1)-A(G_2)$. Also, the algorithm uses a total of
\begin{align*}
& \sum_{j\in[k]}O(\theta^{-2} j 2^{3j}\log(2k/\delta))~\text{calls to}~\tOnaran(G_1,G_2,j,\cdot)\\
& \hspace{3em} = \sum_{j\in[O(1/\e)]}2^{O(1/\e)}~\text{calls to}~\tOnaran(G_1,G_2,j,\cdot).
\end{align*}
In the above equality we also used that $\delta = 0.1$.
\end{proof}

An interesting open problem is whether similar algorithms exist for comparing two graphs on the same vertex set without a common degree matrix $D$.

\end{document}